\documentclass[conference]{IEEEtran}

\usepackage[T1]{fontenc}
\usepackage{amsmath,amssymb,amsthm}
\usepackage{graphicx}
\usepackage{booktabs}
\usepackage{url}
\usepackage{cite}
\usepackage{xcolor}
\usepackage{microtype}
\usepackage{algorithm}
\usepackage{algpseudocode}
\usepackage[acronym,nomain]{glossaries}
\makeglossaries
\newacronym{hwsn}{HWSN}{heterogeneous wireless sensor network}
\newacronym{rid}{RID}{R\'enyi information dimension}
\newacronym{snr}{SNR}{signal-to-noise ratio}
\newacronym{nmse}{NMSE}{normalized mean square error}
\newacronym{awgn}{AWGN}{additive white Gaussian noise}
\newacronym{pf}{PF}{proportional-fair}
\newacronym{isac}{ISAC}{integrated sensing and communication}
\newacronym{ecg}{ECG}{electrocardiogram}
\newacronym{MAC}{MAC}{medium access control}
\newtheorem{theorem}{Theorem}

\newtheorem{remark}{Remark}

\begin{document}

\setlength{\textfloatsep}{1pt}
\setlength{\abovecaptionskip}{1pt}
\setlength{\belowcaptionskip}{1pt}
\setlength{\abovedisplayskip}{3pt}
\setlength{\belowdisplayskip}{3pt}
\allowdisplaybreaks

\title{Power Reduction in Heterogeneous Wireless Sensor Networks via Source-Aware Allocation}

\author{\IEEEauthorblockN{Mauro Marchese\IEEEauthorrefmark{1}, Pietro Savazzi\IEEEauthorrefmark{1}\IEEEauthorrefmark{2}}
\IEEEauthorblockA{\IEEEauthorrefmark{1}Department of Electrical, Computer and Biomedical Engineering, University of Pavia, Italy\\
E-mail: mauro.marchese01@universitadipavia.it}
}

\maketitle

\begin{abstract}
\Glspl{hwsn} in space and extreme environments must reliably transmit diverse analog
physical signals over resource-constrained fading channels, subject
to bandwidth limitations, power budgets, and reconstruction quality
requirements.
This paper addresses two fundamental questions: (i)~what is the minimum
\gls{snr} a sensing link must sustain to reconstruct
an analog signal at a prescribed distortion, regardless of the decoder
used, and (ii)~how can knowledge of the signal's intrinsic structure
be exploited to jointly allocate power and bandwidth across an \gls{hwsn}?
Both questions are answered through the
\gls{rid}, which quantifies the intrinsic
complexity of an analog source distribution.
By combining the \gls{rid} with rate-distortion theory and Shannon channel
capacity, a closed-form \gls{snr} lower bound is derived, parameterized
solely by the source \gls{rid}.
Building on these foundations, a cross-layer resource allocation
framework is introduced that exploits the per-node \gls{rid} to jointly
assign transmit power and bandwidth, achieving strict power saving relative to a Gaussian-assumption baseline while
guaranteeing prescribed reconstruction quality and outage constraints
at every node.
\end{abstract}


\glsresetall
\section{Introduction}
\label{sec:intro}

\Glspl{hwsn} underpin a broad
class of cyber-physical systems, including industrial automation,
structural health monitoring, environmental surveillance, and space applications like remote
physiological monitoring in extreme environments~\cite{elkhediri2022wireless,
Jamshed2022, trigka2025}.
In each of these applications, sensor nodes must reliably transmit
analog physical measurements over low-power radio links
subject to stringent energy budgets, limited spectrum, and fading
propagation~\cite{Karimi2024,Tossa2025}.
The battery lifetime of each node is the primary design bottleneck:
once a node exhausts its energy reserve, the sensing coverage it
provided is permanently lost, degrading the network's ability to
monitor the target phenomenon~\cite{Sahoo2024}.

A substantial body of literature addresses the energy-efficiency
problem in \glspl{hwsn} from multiple angles.
At the network layer, clustering protocols organize
nodes into groups so that only cluster-heads forward aggregated data,
dramatically reducing the number of transmissions and extending network
lifetime~\cite{Sahoo2024, elkhediri2022wireless}.
At the protocol layer, energy-efficient communication schemes
design adaptive duty cycles, low-overhead routing, and resilient
data-aggregation mechanisms to minimize idle listening and
retransmissions~\cite{Dhabliya2022}.
At the physical and access layer, advanced multiple
access techniques such as non-orthogonal multiple access (NOMA) allow multiple
nodes to share the same spectrum with differentiated power levels,
improving spectral and energy efficiency for IoT-dense
deployments~\cite{Raj2022}.
Machine-learning-based approaches have also been proposed to optimize
energy consumption by learning adaptive transmission and sleep
schedules from historical traffic patterns~\cite{Bagwari2023}.
Despite their diversity, these methodologies share a common implicit
assumption: the power budget required to maintain a reliable link to
the gateway is treated as signal-independent quantity and no source-aware approach has been investigated so far.

From an information-theoretic standpoint, the fundamental question is:
\emph{what is the minimum \gls{snr} that a sensing
link must sustain to reconstruct an analog measurement at a prescribed
quality, regardless of the decoder used?}
A similar problem has been investigated in \cite{gastpar2005}, where a lower bound on the achievable distortion has been  derived under Gaussian source assumptions.
Extending such bounds to arbitrary source distributions for source-aware resource allocation requires a
measure of intrinsic source complexity that generalizes beyond
Gaussianity.

The \gls{rid}~\cite{renyi1959dimension,
wu2010renyi} provides exactly this measure: it quantifies the
effective number of degrees of freedom per sample of an analog
source distribution, independently of any specific model or
reconstruction algorithm.
Through the Kawabata--Dembo theorem~\cite{kawabata1994rate}, the \gls{rid}
governs the high-resolution rate-distortion behavior of a source,
making it the natural parameter for expressing fundamental
communication limits in terms of signal structure rather than
Gaussian assumptions.
Sources with \gls{rid} strictly below unity, as is the case for structured physical signals, can be reliably reconstructed
at a strictly lower \gls{snr} than Gaussian sources at the same target distortion, implying that the power budgets derived under Gaussian
assumptions are systematically inflated for most \gls{hwsn} nodes.
Exploiting this structural advantage at the resource allocation level
requires a cross-layer framework that connects source-aware power and bandwidth assignment at the physical layer, an
approach that has not been
investigated in the context of \glspl{hwsn}.

In light of the above discussion, this work is aimed to develop a comprehensive cross-layer resource allocation framework for \glspl{hwsn} that leverages the \gls{rid} to transcend traditional Gaussian assumptions and optimize energy budgets based on actual source complexity.
The contributions of this work are:
\begin{itemize}
  \item \textbf{Fundamental \gls{snr} bound}: A closed-form lower bound on the \gls{snr} required to reconstruct an arbitrary analog source with achievable \gls{nmse} distortion~$D$,
        parameterized solely by the source \gls{rid}, which quantifies the intrinsic complexity of the source. This bound allows for the design of novel resource allocation schemes for \gls{hwsn} and provides insights for the deployment of sensor nodes (maximum sensing distance).  
  \item \textbf{\Gls{rid}-aware cross-layer resource allocation framework}: A novel power and bandwidth allocation framework, accounting for the intrinsic source complexity (\gls{rid}), has been developed. The design introduces the concept of
        \emph{power spectral cost}, which relates power and bandwidth and allows for simplifying the optimization problem. Simulation results confirm that the proposed scheduling algorithm effectively reduces power consumption in \gls{hwsn}, with a power saving of almost $4.5$ dB compared to the source-ignorant Gaussian baseline.
\end{itemize}

\section{System Model and Background}
\label{sec:model}

\subsection{Signal Acquisition and Channel Model}

Consider a \gls{hwsn} in which $K$ sensor nodes transmit analog
measurements to a common gateway.
Node $k$ measures a scalar analog zero mean source $X_k \sim p_{X_k}$ with
power $P_{k} = \mathbb{E}[X_k^2]$ and transmits it over a
Rayleigh flat-fading channel.
The received signal at the gateway is modeled as
\begin{equation}
  Y_k = \sqrt{h_k}\,X_k + N_k,
  \label{eq:channel}
\end{equation}
where $h_k$ is the instantaneous channel power gain and
$N_k\sim\mathcal{N}(0,\sigma_k^2)$ is \gls{awgn}
with variance $\sigma_k^2 = N_0 B_k$, with $N_0$ being the
one-sided noise spectral density and $B_k$ the bandwidth allocated
to node~$k$.
Under Rayleigh fading, the channel gain is exponentially
distributed as $h_k \sim (1/\bar{h}_k)\mathrm{exp}(1/\bar{h}_k)$, where
$\bar{h}_k = \mathbb{E}[h_k]$ is the mean channel gain.

The mean channel gain follows the log-distance path-loss
model
\begin{equation}
  \bar{h}_k = h_0\!\left(\frac{r_0}{r_k}\right)^{\!n},\quad
  h_0 = \left(\frac{\lambda}{4\pi r_0}\right)^{\!2},
  \label{eq:pathloss}
\end{equation}
where $r_k$ is the sensor-to-gateway distance, $r_0$ is a reference
distance, $n \geq 2$ is the path-loss exponent, and $\lambda$ is the
carrier wavelength.
The instantaneous \gls{snr} at the gateway for node $k$ is
$\mathrm{SNR}_k = h_k P_k / (N_0 B_k)$, where $P_k$ is the transmit
power allocated to node~$k$.
The corresponding mean \gls{snr} is
\begin{equation}
  \overline{\mathrm{SNR}}_k = \frac{\bar{h}_k P_k}{N_0 B_k}.
  \label{eq:mean_snr}
\end{equation}
At the gateway, a decoder $\hat{X}_k = g_k(Y_k)$ reconstructs the source with
\gls{nmse} distortion
\begin{equation}
  D_k = \frac{\mathbb{E}[(X_k - \hat{X}_k)^2]}{P_{k}}.
  \label{eq:nmse}
\end{equation}

\subsection{The R\'enyi Information Dimension}

The \gls{rid} provides a characterization of the intrinsic complexity of an analog
source distribution~\cite{renyi1959dimension, wu2010renyi}.
Let $\langle X \rangle_m = \lfloor mX \rfloor / m$ denote the
$1/m$-quantized version of a source $X$.
The \gls{rid} of $X$ is defined as
\begin{equation}
  d(X) = \lim_{m \to \infty} \frac{H(\langle X \rangle_m)}{\log_2 m},
  \label{eq:rid}
\end{equation}
where $H(\cdot)$ denotes Shannon entropy.
Specifically,
\begin{equation}
    H(\langle X\rangle_m)= -\sum_{\ell=L_{min}}^{L_{max}} p_{\ell}(m) \log_2 p_{\ell}(m),
\end{equation}
where 
\begin{equation}
p_{\ell}(m) = P\left( \frac{\ell}{m} \leq X < \frac{\ell+1}{m} \right)= \int_{\ell/m}^{(\ell+1)/m} p_X \, dx,    
\end{equation} 
and $L_{min} = \lfloor m x_{\min} \rfloor$, $L_{max} = \lfloor m x_{\max} \rfloor$ define the range of indices covering the support $[x_{\min}, x_{\max}]$.
Intuitively, $d(X)$ measures the rate at which the quantization entropy
grows as the resolution is refined.
The \gls{rid} satisfies $d(X) \in [0,1]$ for scalar sources and
it can be directly estimated from data by evaluating $H(\langle X
\rangle_m)$ for increasing values of $m$ and measuring the asymptotic slope.
Importantly, the \gls{rid} is a property of the source distribution,
independent of any specific compression algorithm or channel model.

Under the assumption of high resolution, or small distortion, the rate-distortion function of $X$ can be related to the \gls{rid} and the rate-distortion function of gaussian sources. This result is summarized in the following theorem.

\begin{theorem}[Kawabata--Dembo~\cite{kawabata1994rate}]
\label{thm:kd}
Let $X$ have finite \gls{rid} $d(X)$ and let the distortion measure be the
\gls{nmse}. Then, as $D \to 0$,
\begin{equation}
  R(D) = \frac{d(X)}{2}\log_2\frac{1}{D}
        + o\!\left(\log_2\frac{1}{D}\right).
  \label{eq:kd}
\end{equation}
\end{theorem}

Theorem~\ref{thm:kd} generalizes the classical Gaussian
formula $R_{\mathrm{G}}(D) = \frac{1}{2}\log_2(1/D)$ (corresponding
to $d(X)=1$) to arbitrary source distributions.
The fundamental implication is that, for a given target distortion~$D$,
a source with $d(X) < 1$ requires only $d(X)$ times as many bits per
sample as a Gaussian source of the same power
\begin{equation}
R(D)\approx  d(X) R_{\mathrm{G}}(D) . 
\end{equation}
This structural advantage is the information-theoretic foundation of the results derived in the following sections.

\section{Fundamental SNR Bound}
\label{sec:bound}
In this section, the minimum \gls{snr} required by any decoder to reconstruct $X$ from the noisy observation $Y$ with \gls{nmse} target distortion $D$ is derived.

\subsection{Derivation}
The derivation combines Th.~\ref{thm:kd} with the \gls{awgn} channel capacity. The key result is summarized in the following theorem.

\begin{theorem}
\label{thm:snrbound}
Under the channel model~\eqref{eq:channel}, a necessary condition for
reconstructing $X$ with achievable distortion $D$ (for $D \ll 1$) is
\begin{equation}
  \mathrm{SNR} \;\geq\; \mathrm{SNR}_{\min}(D) = D^{-d(X)} - 1.
  \label{eq:snrbound}
\end{equation}
Conversely, for $\mathrm{SNR} > \mathrm{SNR}_{\min}(D)$, there exists a
(possibly computationally complex) decoder that achieves distortion~$D$.
\end{theorem}

\begin{proof}
Reconstruction with distortion $D$ requires extracting at least $R(D)$
bits of information about $X$ from $Y$, i.e.,
$I(X; Y) \geq R(D)$.
Since the mutual information is upper bounded by the \gls{awgn} channel
capacity,
\begin{equation}
  I(X; Y) \leq C_{\mathrm{AWGN}} = \tfrac{1}{2}\log_2(1 + \mathrm{SNR}),
\end{equation}
the necessary condition $R(D) \leq C_{\mathrm{AWGN}}$, combined with
Theorem~\ref{thm:kd}, yields
\begin{equation}
  \frac{d(X)}{2}\log_2\frac{1}{D}
  \;\leq\; \frac{1}{2}\log_2(1 + \mathrm{SNR}).
\end{equation}
Exponentiating both sides gives $D^{-d(X)} \leq 1 + \mathrm{SNR}$,
which is~\eqref{eq:snrbound}.
Achievability follows from the achievability of $C_{\mathrm{AWGN}}$.
\end{proof}

\begin{remark}
The bound~\eqref{eq:snrbound} is algorithm-agnostic: it holds
regardless of the reconstruction algorithm deployed at the receiver.
It therefore defines the informational boundary beyond which no
algorithm, however sophisticated, can succeed. Moreover, in case of Gaussian sources, \eqref{eq:snrbound} reduces to 
\begin{equation}\label{eq:Dverdu}
    D=\frac{1}{1+\text{SNR}_{\min}},
\end{equation}
which corresponds to \cite[Eq. 123]{wu2011mmse}. Therefore, the bound in \eqref{eq:snrbound} serves as a generalization of \eqref{eq:Dverdu}.
\end{remark}

Moreover, for small distortion ($D \ll 1$), $\mathrm{SNR}_{\min}(D) \approx D^{-d(X)}$,
or equivalently in dB,
\begin{equation}
  \mathrm{SNR}_{\min}^{(\mathrm{dB})}(D) \approx d(X)  |D^{(\mathrm{dB})}|.
  \label{eq:snr_db}
\end{equation}
Thus, $d(X)$ acts as a \emph{scaling factor} on the required \gls{snr} in
logarithmic scale.
A source with $d(X) < 1$ enjoys a direct, quantifiable \gls{snr} advantage
over a Gaussian source.

\subsection{Pareto Frontiers in the SNR--Distortion Plane}

Equation~\eqref{eq:snrbound} defines a family of Pareto frontiers in
the $(\mathrm{SNR}_{\min}, D)$ plane, one for each value of $d(X) \in [0, 1]$.
These curves are shown in Fig. \ref{fig:paretofrontiers}, and each curve represents a fundamental limit: the region below it is
unachievable by any decoder.
The Gaussian source ($d(X) = 1$) constitutes the worst case, requiring
the highest minimum \gls{snr} to achieve any prescribed distortion.
As $d(X)$ decreases from unity toward zero, the required minimum \gls{snr}
decreases monotonically for every fixed distortion target, reflecting
the increasing structural simplicity of the source.

\begin{figure}[!t]
    \centering
    \includegraphics[width=\columnwidth]{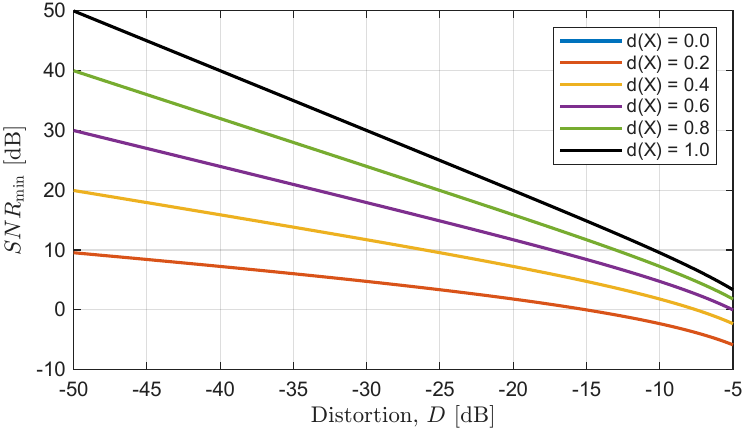}
    \caption{Minimum SNR required ($SNR_{\min}$) as a function of the target NMSE distortion $D$ for different values of the source R{\'e}nyi information dimension $d(X)$. The black curve ($d(X) = 1.0$) represents the classical Gaussian-assumption baseline, which constitutes the worst-case resource requirement. As the source complexity decreases ($d(X) \rightarrow 0$), the Pareto frontier shifts downward, demonstrating a direct and quantifiable SNR advantage for structurally simpler analog sources at any distortion level.}
    \label{fig:paretofrontiers}
\end{figure}

\section{Proposed Cross-Layer Resource Allocation Framework}
\label{sec:crosslayer}
This section introduces the proposed source-aware resource allocation framework for \glspl{hwsn}, based on the model in Sec. \ref{sec:model} and the fundamental \gls{snr} bound in Th. \ref{thm:snrbound}.

\subsection{Outage Probability under Rayleigh Fading}

Given the \gls{snr} requirement in \eqref{eq:snrbound}, an outage event for node $k$ is defined as the event in which
the instantaneous \gls{snr} falls below the fundamental threshold
\begin{equation}
  \mathcal{O}_k =
  \bigl\{\mathrm{SNR}_k < \mathrm{SNR}_{\min}(D_k)\bigr\}.
  \label{eq:outage_event}
\end{equation}
Since, under Rayleigh fading the instantaneous \gls{snr} is exponentially distributed with mean $\overline{\mathrm{SNR}}_k$, giving an outage probability of $P_{\mathrm{out}}^{(k)} = 1 - \exp(-\mathrm{SNR}_{\min}(D_k)/\overline{\mathrm{SNR}}_k)$. Substituting~\eqref{eq:snrbound} yields
\begin{equation}
  P_{\mathrm{out}}^{(k)} = 1 - \exp\!\left(
    -\frac{D_k^{-d_k} - 1}{\,\overline{\mathrm{SNR}}_k\,}
  \right).
  \label{eq:pout_ray}
\end{equation}
where the per-node target distortion is $D_k$ and the source \gls{rid} is $d_k$.
Equation~\eqref{eq:pout_ray} shows that for a fixed mean \gls{snr},
the outage probability decreases monotonically as $d_k$ decreases,
since the threshold $D_k^{-d_k} - 1$ is a decreasing function of
$d_k$ for any $D_k < 1$.
This confirms that the structural advantage of low-\gls{rid} sources
carries over to the fading regime.

Given an outage constraint $P_{\mathrm{out}}^{(k)}\leq\varepsilon_k \in (0,1)$, the required mean
\gls{snr} is obtained by inverting~\eqref{eq:pout_ray}
\begin{equation}
  \overline{\mathrm{SNR}}_k(\varepsilon_k)
  = \frac{D_k^{-d_k} - 1}{\ln\bigl(1/(1-\varepsilon_k)\bigr)}.
  \label{eq:snravg_req}
\end{equation}

Substituting \eqref{eq:pathloss} and \eqref{eq:mean_snr} into~\eqref{eq:snravg_req} and solving for
$r_k$ yields the maximum sensor-to-gateway distance for node $k$ achieving outage
$P_{\mathrm{out}}^{(k)} \leq \varepsilon_k$ with transmit power $P_k$ and
bandwidth $B_k$ is
\begin{equation}
  r_k^{\max} = r_0\left(
    \frac{h_0\,P_k\,\ln\!\bigl(1/(1-\varepsilon_k)\bigr)}
         {N_0 B_k\,(D_k^{-d_k}-1)}
  \right)^{\!\!1/n}.
  \label{eq:rmax}
\end{equation}
Three structural properties follow from~\eqref{eq:rmax}.
First, $r_k^{\max}$ is a decreasing function of $d_k$: structurally
simpler sources can be reliably acquired at greater distances under the
same link budget.
Second, $r_k^{\max} \propto P_k^{1/n}$ indicates that transmit power and range
trade off at a rate governed by the propagation environment.
Third, $r_k^{\max}$ is increasing in $\varepsilon_k$: relaxing the
outage constraint permits longer sensing links at the cost of more
frequent reconstruction failures. In conclusion, \eqref{eq:rmax} serves as a \gls{rid}-aware link-budget equation for the design of \glspl{hwsn}.

\subsection{Source-Aware Allocation: Problem Formulation}
This section formulates the proposed resource allocation framework.
Consider the gateway allocating transmit power $P_k$ and bandwidth
$B_k$ to each node $k \in \mathcal{K} = \{1,\ldots,K\}$, subject to
a total bandwidth $B_{\mathrm{tot}}$ and a total power budget
$P_{\mathrm{tot}}$.
Each node $k$ is characterized by its source \gls{rid} $d_k$, distortion
target $D_k$, distance $r_k$, and outage budget $\varepsilon_k$.
With the objective of minimizing network power consumption, the resource allocation problem is formulated as
\begin{align}
(\text{P1}): \quad \min_{\{P_k,\,B_k\}} \quad & \sum_{k=1}^K P_k \label{eq:opt_obj} \\
\text{s.t.} \quad & P_{\mathrm{out}}^{(k)}(P_k, B_k) \leq \varepsilon_k, \quad \forall k, \label{eq:out_constr} \\
& \sum_{k=1}^K B_k \leq B_{\mathrm{tot}}, \label{eq:bw_constr} \\
& \sum_{k=1}^K P_k \leq P_{\mathrm{tot}}, \label{eq:pwr_constr} \\
& P_k\geq 0, B_k \geq B_k^{\min}, \quad \forall k, \label{eq:pos_constr}
\end{align}
where $B_k^{\min} > 0$ is a minimum bandwidth.
Under Rayleigh fading, constraint~\eqref{eq:out_constr} is equivalent,
via \eqref{eq:snravg_req}, to
\begin{equation}
  \overline{\mathrm{SNR}}_k
  = \frac{\bar{h}_k\,P_k}{N_0 B_k}
  \;\geq\;
  \frac{D_k^{-d_k} - 1}{\ln\!\bigl(1/(1-\varepsilon_k)\bigr)},
  \label{eq:snr_constr}
\end{equation}
which is a linear constraint in the ratio $P_k/B_k$.

\subsection{Power Spectral Cost and Closed-Form Policy}
In order to provide a simple allocation scheme, it can be noted that,
activating constraint~\eqref{eq:snr_constr} with equality, the
minimum transmit power for node $k$ with a given bandwidth $B_k$ is
\begin{equation}
  P_k^*(B_k) = \pi_k  B_k,
  \label{eq:pk_star}
\end{equation}
where the \emph{power spectral cost} $\pi_k$ of node $k$ is defined as
\begin{equation}
  \pi_k =
  \frac{N_0\,(D_k^{-d_k}-1)}
       {\bar{h}_k\,\ln\!\bigl(1/(1-\varepsilon_k)\bigr)}.
  \label{eq:price}
\end{equation}
The power spectral cost $\pi_k$ encapsulates all node-specific
parameters into a single scalar that quantifies the power cost of
allocating one unit of bandwidth to node $k$.
Crucially, $\pi_k$ depends on the physical nature of the sensed
process: nodes reporting
low-\gls{rid} signals have smaller $\pi_k$ and can therefore be served
more efficiently.

\begin{remark}
A \gls{rid}-\emph{unaware} (Gaussian-assumption) baseline assigns each node
$\pi_k^{\mathrm{G}} = N_0(D_k^{-1}-1)/[\bar{h}_k\ln(1/(1-\varepsilon_k))]$,
obtained by setting $d_k = 1$ in~\eqref{eq:price}. This corresponds to source-ignorant resource allocation scheme.
The per-node power saving factor of a \gls{rid}-aware policy is
\begin{equation}
  \eta_k^P
  = \frac{\pi_k^{\mathrm{G}}}{\pi_k}
  = \frac{D_k^{-1} - 1}{D_k^{-d_k} - 1} \;>\; 1
  \quad \forall\;d_k < 1,
  \label{eq:eta_power}
\end{equation}
which is strictly greater than one for any source with $d_k < 1$. 
\end{remark}
Since $P_k^*(B_k)$ is linear in $B_k$, the total power consumed by
a bandwidth allocation $\{B_k\}_{k=1}^K$ is
$\sum_{k=1}^K \pi_k B_k$.
Minimizing this over feasible bandwidth allocations subject
to~\eqref{eq:bw_constr}--\eqref{eq:pos_constr} decouples into
\begin{align}
(\text{P2}): \quad \min_{\{B_k \geq B_k^{\min}\}} \quad & \sum_{k=1}^K \pi_k B_k \label{eq:bw_opt_obj} \\
\text{s.t.} \quad & \sum_{k=1}^K B_k \leq B_{\mathrm{tot}}. \label{eq:bw_opt_constr}
\end{align}
Problem P2 is a linear program with a closed-form
solution: assign the minimum floor $B_k^{\min}$ to all nodes, then
allocate all remaining bandwidth to the node with the smallest $\pi_k$
\begin{equation}
  B_k^* = B_k^{\min}
  + \Big(B_{\mathrm{tot}} - \sum_{j=1}^K B_j^{\min}\Big)\delta_{kk^*},
  \label{eq:bw_policy}
\end{equation}
where $k^* = \arg\min_k\, \pi_k$.
The corresponding power for each node is then obtained as $P_k^* = \pi_k B_k^*$.

\begin{algorithm}[t]
  \caption{\gls{rid}-Aware Cross-Layer Resource Allocation}
  \label{alg:rid_alloc}
  \begin{algorithmic}[1]
    \Require Active nodes $\mathcal{K}$;
             total bandwidth $B_{\mathrm{tot}}$;
             total power budget $P_{\mathrm{tot}}$;
             per-node parameters $\{d_k,D_k,\,\varepsilon_k,\,r_k,\,
             B_k^{\min}\}_{k \in \mathcal{K}}$
    \Ensure Per-node allocations $\{P_k,\, B_k\}_{k \in \mathcal{K}}$
    \State Compute excess bandwidth $B_{\mathrm{exc}}$
    \For{each active node $k \in \mathcal{K}$}
      \State Compute $\pi_k$ using \eqref{eq:price}
      \State Assign bandwidth following \eqref{eq:bw_pf}
      \State Compute required power using \eqref{eq:pk_star}
    \EndFor
    \If{$\sum_k P_k > P_{\mathrm{tot}}$}
      \State Compute scaling factor: $\alpha = P_{\mathrm{tot}} / \sum_k P_k$
      \For{each node $k \in \mathcal{K}$}
        \State Rescale power as $P_k \leftarrow \alpha P_k$
      \EndFor
    \EndIf
    \State Transmit allocation $(P_k,\, B_k)$ to each node~$k$
  \end{algorithmic}
\end{algorithm}

\subsection{Proportional-Fair Bandwidth Allocation}

When concentrating all excess bandwidth at a single node is
undesirable, a \gls{pf} variant is preferred. To ensure the minimum floor is respected, we first allocate $B_k^{\min}$ to each node and then distribute the excess bandwidth $B_{\mathrm{exc}} = B_{\mathrm{tot}} - \sum_j B_j^{\min}$ inversely proportional to the power spectral costs
\begin{equation}
  B_k^{\mathrm{PF}} = B_k^{\min} + B_{\mathrm{exc}}  \frac{1/\pi_k}{\sum_j 1/\pi_j}.
  \label{eq:bw_pf}
\end{equation}
Under this policy, nodes with lower $\pi_k$---i.e., nodes reporting
more structured, low-\gls{rid} signals---are assigned wider bandwidth. Note that in highly heterogeneous scenarios, nodes with very high $\pi_k$ will effectively receive a total bandwidth very close to $B_k^{\min}$.

\subsection{RID-Aware Cross-Layer Scheduling Algorithm}

Algorithm~\ref{alg:rid_alloc} summarizes the proposed cross-layer
resource allocation policy.
At each scheduling interval, the gateway
(i)~collects the \gls{rid} of each active node,
(ii)~computes the power spectral costs $\{\pi_k\}$,
(iii)~assigns bandwidth via the \gls{pf}
rule~\eqref{eq:bw_pf}, and
(iv)~informs each node of its allocated $B_k$ and power $P_k$.

The feasibility condition for the power budget constraint is
$\sum_{k \in \mathcal{K}} \pi_k B_k^{\min} \leq P_{\mathrm{tot}}$.
If the total power after allocation exceeds $P_{\mathrm{tot}}$, the algorithm applies a uniform rescaling factor $\alpha < 1$. This ensures that the power budget is respected while maintaining the proportional fairness of the bandwidth distribution, although it results in an increase in the effective achievable distortions $D_{k,\mathrm{eff}}$.
The \gls{rid}-aware policy minimizes the occurrence of this rescaling by
reducing $\pi_k$ for structured-signal nodes.

\begin{table}[t]
\renewcommand{\arraystretch}{1.25}
\caption{Simulation Parameters}
\label{tab:sim_params}
\centering
\begin{tabular}{lll}
\hline
\textbf{Parameter} & \textbf{Symbol} & \textbf{Value} \\
\hline
Carrier frequency          & $f_c$            & 2.4\,GHz         \\
Path-loss exponent         & $n$              & 3.5              \\
Reference distance         & $d_0$            & 1\,m             \\
Total bandwidth            & $B_\mathrm{tot}$ & 2\,MHz           \\
Minimum bandwidth fraction & $B_\mathrm{min}/B_\mathrm{tot}$ & 0.1 (200\,kHz) \\
Receiver noise figure      & $\mathrm{NF}$    & 7\,dB            \\
Reference temperature      & $T_0$            & 290\,K           \\
Target distortion           & $D_k$            & $-25$\,dB        \\
Outage probability budget  & $\varepsilon_k$  & 0.1              \\
Total power budget         & $P_\mathrm{tot}$ & 10\,dBm          \\
Number of nodes            & $K$              & 8                \\
\hline
\end{tabular}
\end{table}

\begin{table}[t]
\renewcommand{\arraystretch}{1.25}
\caption{Per-Node Configuration}
\label{tab:node_params}
\centering
\setlength{\tabcolsep}{5pt}
\begin{tabular}{lcccccccc}
\hline
\textbf{Node} & N1 & N2 & N3 & N4 & N5 & N6 & N7 & N8 \\
\hline
RID, $d_k$ & 0.5 & 0.65 & 0.92 & 0.5 & 0.65 & 0.92 & 0.65 & 0.5 \\
Distance, $r_k$ [m] & 5 & 8 & 10 & 15 & 20 & 25 & 12 & 18 \\
\hline
\end{tabular}
\end{table}

\begin{figure}[!t]
    \centering
    \includegraphics[width=\columnwidth]{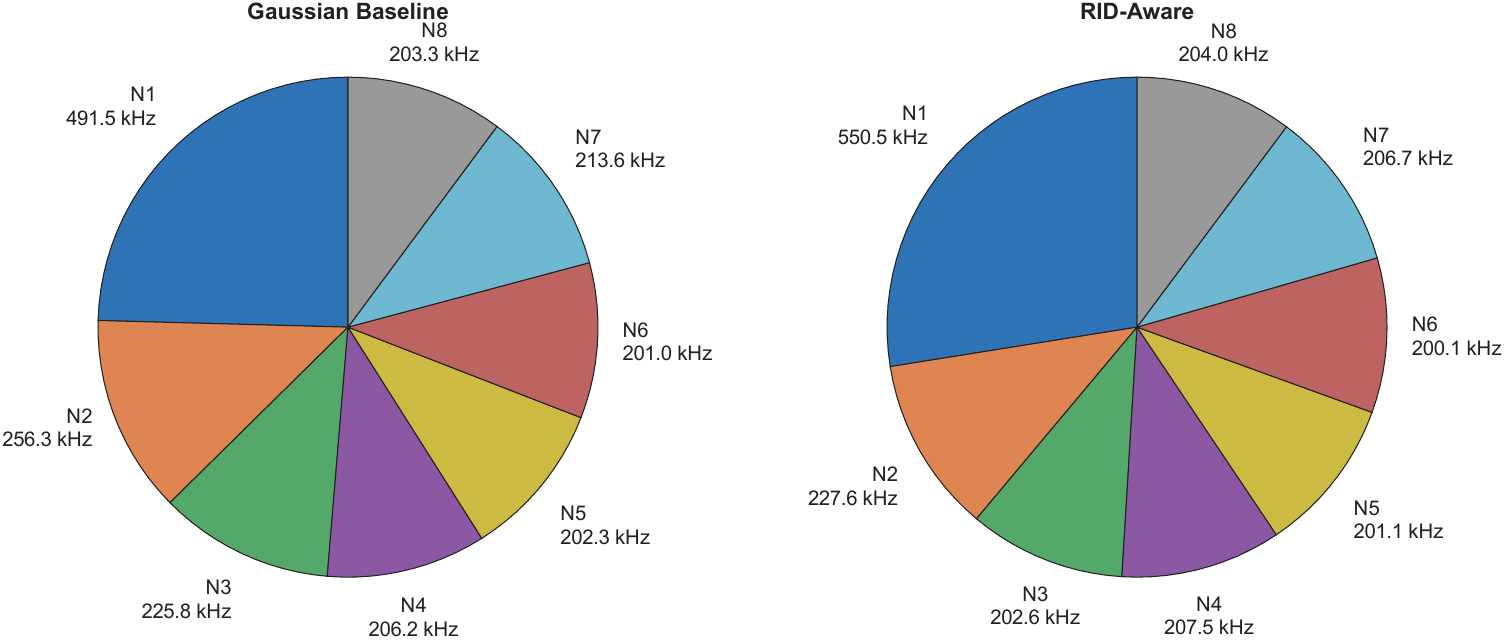}
    \caption{Bandwidth allocation comparison between the Gaussian baseline (left) and the proposed \gls{rid}-aware \gls{pf} policy (right). The \gls{rid}-aware approach adaptively redistributes bandwidth based on the structural complexity of each node's signal.}
    \label{fig:BWallocation}
\end{figure}

\section{Simulation Results}

\subsection{Simulation Scenario}

The proposed \gls{rid}-aware cross-layer resource allocation framework is
evaluated on an \gls{hwsn} consisting of $K = 8$ sensor nodes
transmitting to a common gateway over independent Rayleigh fading
channels. The nodes report physically distinct signal classes and sensor-to-gateway distances range from 5\,m to 25\,m,
covering a representative indoor/short-range industrial deployment
scenario. All nodes share a uniform target distortion of
$D_k = -25$\,dB and an outage budget of $\varepsilon_k = 0.1$.
The radio link operates at a carrier frequency of $f_c = 2.4$\,GHz
with a path-loss exponent $n = 3.5$, modeling an obstructed
propagation environment. The total available bandwidth is
$B_\mathrm{tot} = 2$\,MHz and the total transmit power budget is
$P_\mathrm{tot} = 10$\,dBm. Each node is guaranteed a minimum
bandwidth of $B_\mathrm{min} = 200$\,kHz.
All simulation parameters are summarized in
Table~\ref{tab:sim_params}, and the per-node configuration is detailed
in Table~\ref{tab:node_params}.

The proposed \gls{rid}-aware \gls{pf} policy (Algorithm~1) is
compared against a source-ignorant Gaussian-assumption baseline, which sets $d_k =
1$ for all nodes regardless of the true signal structure, and applies
an equivalent \gls{pf} bandwidth allocation under the same power and outage constraints.

\begin{figure}[!t]
    \centering
    \includegraphics[width=\columnwidth]{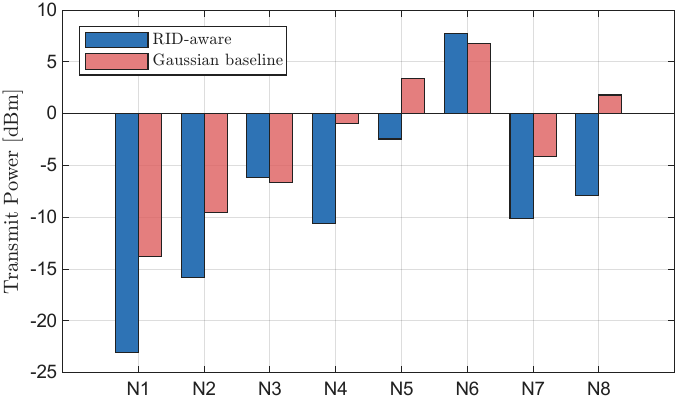}
    \caption{Transmit power levels for nodes N1 through N8. The \gls{rid}-aware policy (blue) achieves significant power savings compared to the \gls{rid}-unaware Gaussian baseline (red) by exploiting the intrinsic information dimension of the sensed signals.}
    \label{fig:PowerAllocation}
\end{figure}

\begin{figure}[!t]
    \centering
    \includegraphics[width=\columnwidth]{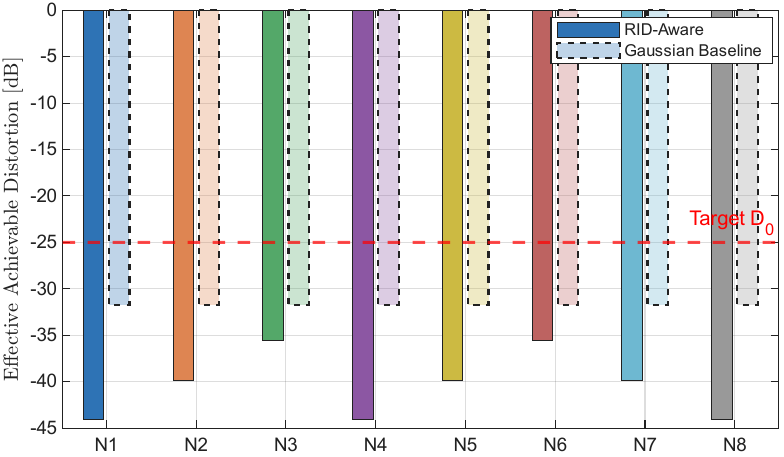}
    \caption{Effective achievable distortion across the sensor nodes. While both policies aim to meet the target distortion threshold (red dashed line), the \gls{rid}-aware framework provides a superior distortion margin for nodes with structured (low-\gls{rid}) source distributions.}
    \label{fig:effDistortion}
\end{figure}

\subsection{Results and Discussion}

Fig.~\ref{fig:BWallocation} compares the bandwidth allocated to each
node under the Gaussian baseline and the proposed \gls{rid}-aware \gls{pf} policy.
Under the Gaussian baseline, the allocation is primarily driven by
path-loss differences, with N1 (the closest node) receiving the largest
share ($\approx 491.5$\,kHz) due to its favorable channel conditions.
The \gls{rid}-aware policy modifies this distribution by additionally
accounting for each node's intrinsic signal complexity: N1, which
reports a low-\gls{rid} ($d = 0.50$), is further
rewarded with an even larger bandwidth share ($\approx 550.5$\,kHz),
while nodes reporting higher-\gls{rid} (N3, N6) receive
allocations close to their minimum floor. This behavior is consistent
with the physical rationale underlying the \gls{pf} rule~\eqref{eq:bw_pf}: lower-$\pi_k$
nodes benefit more from additional bandwidth increments, yielding
higher power efficiency per allocated Hz.


Fig.~\ref{fig:PowerAllocation} shows the transmit power assigned to
each node under both policies. The \gls{rid}-aware allocation systematically
reduces the required power at every node relative to the Gaussian
baseline, with the largest savings observed at nodes reporting
structured low-\gls{rid} signals (N1, N4, N8). This is a
direct consequence of the reduced power spectral cost $\pi_k$ for
low-\gls{rid} sources, as formalized in Remark~3 and~\eqref{eq:price}. The per-node power saving factor $\eta_k^P$ quantifies this advantage in closed form: for
$D_k = -25$\,dB and $d_k = 0.50$, the theoretical saving exceeds
10\,dB, consistent with the numerical results shown.


Fig.~\ref{fig:effDistortion} reports the effective achievable
distortion at each node after resource allocation. Both policies target
the same distortion threshold $D_k = -25$\,dB (red dashed line). The
\gls{rid}-aware framework not only meets the target at all nodes, but also
provides a substantially larger distortion margin for nodes with
low-\gls{rid} source distributions. This margin reflects the fact that, at
the allocated power and bandwidth levels, structured-signal nodes
operate at an \gls{snr} significantly above the minimum required threshold,
leaving a comfortable safety margin against channel variability. The
Gaussian baseline, by contrast, allocates resources as if all signals
were equally complex, resulting in tighter margins at several nodes. These results confirm
that exploiting source structure via the \gls{rid} allows the gateway to
jointly improve distortion performance and resource efficiency.

\begin{figure}[!t]
    \centering
    \includegraphics[width=\columnwidth]{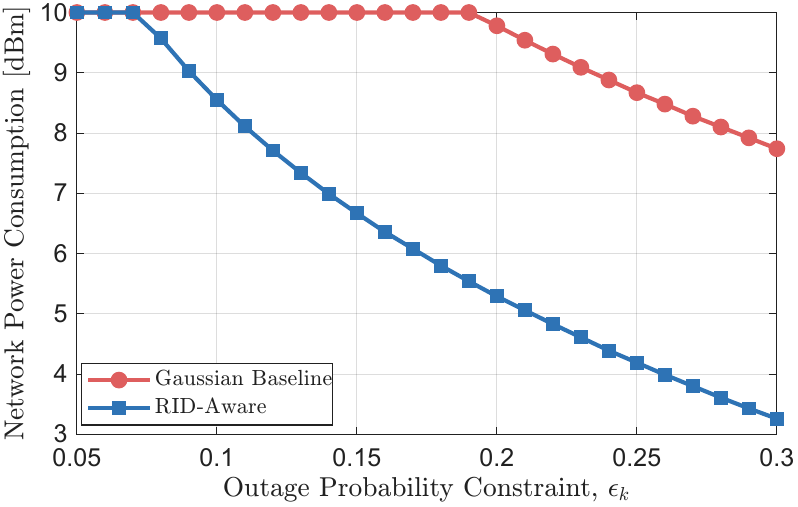}
    \caption{Total transmit power versus the outage probability constraint $\varepsilon_k$. The \gls{rid}-aware allocation consistently maintains a lower energy footprint than the Gaussian baseline across all reliability requirements, demonstrating its efficiency in fading environments.}
    \label{fig:PowerVSoutage}
\end{figure}


Finally, Fig.~\ref{fig:PowerVSoutage} illustrates the total transmit power
consumed by the network as a function of the outage probability
constraint $\varepsilon_k$. As expected, tighter outage requirements (smaller
$\varepsilon_k$) drive up the total power demand for both policies,
since the required mean \gls{snr}~\eqref{eq:snravg_req} grows as $\varepsilon_k \to 0$.
Crucially, the \gls{rid}-aware allocation maintains a consistently lower
total power than the Gaussian baseline across the entire range of
outage constraints, demonstrating that the structural advantage of the
\gls{rid}-aware policy is preserved under varying reliability requirements. Moreover, saturation to $P_{\text{tot}}$ occurs at much lower values of the outage budget compared to the Gaussian baseline.

\section{Conclusion}
This paper introduces a novel RID-aware cross-layer  resource allocation framework for HWSNs in space and extreme environments, shifting the paradigm from traditional Gaussian assumptions to a source-aware approach. By bridging the Rényi information dimension with information theory over Rayleigh fading channels, we derive a fundamental, algorithm-agnostic SNR bound and introduce the concept of power spectral cost. The proposed proportional-fair scheduling policy yields a significant power saving factor, which grows dramatically as distortion constraints tighten. Numerical simulations validate these findings, demonstrating total power reductions of up to 4.5~dB compared to the baseline while guaranteeing target reconstruction quality and outage constraints.

Future work will address the sub-optimality of the current framework. Specifically, the conservative distortion margins of low-RID sources can be exploited to further refine resource allocation and maximize network energy efficiency. Finally, extending this source-aware approach to integrated sensing and communication (ISAC) systems represents a promising research direction.

\bibliographystyle{IEEEtran}
\bibliography{references}

\end{document}